\documentclass[conference]{IEEEtran}
\usepackage[T1]{fontenc}
\usepackage[latin9]{inputenc}
\usepackage{units}
\usepackage{amsthm}
\usepackage{amsmath}
\usepackage{amssymb}
\usepackage{graphicx}

\makeatletter
\theoremstyle{plain}
\newtheorem{thm}{\protect\theoremname}
\theoremstyle{plain}
\newtheorem{prop}[thm]{\protect\propositionname}
\theoremstyle{definition}
\newtheorem{defn}[thm]{\protect\definitionname}
\theoremstyle{plain}
\newtheorem{lem}[thm]{\protect\lemmaname}
\theoremstyle{plain}
\newtheorem{cor}[thm]{\protect\corollaryname}

\usepackage{cite}

\usepackage{balance}

\usepackage{amsthm}

\AtBeginDocument{
\addtolength{\abovedisplayskip}{-1.5ex}
\addtolength{\abovedisplayshortskip}{-1ex}
\addtolength{\belowdisplayskip}{-0.5ex}
\addtolength{\belowdisplayshortskip}{-0.5ex}
}

\makeatother

\providecommand{\corollaryname}{Corollary}
\providecommand{\definitionname}{Definition}
\providecommand{\lemmaname}{Lemma}
\providecommand{\propositionname}{Proposition}
\providecommand{\theoremname}{Theorem}

\begin{document}

\title{User Capacity of Pilot-Contaminated\\
TDD Massive MIMO Systems}

\author{\IEEEauthorblockN{Juei-Chin Shen, Jun Zhang, and Khaled Ben Letaief, \textit{Fellow}, \textit{IEEE}}
\IEEEauthorblockA{Dept. of ECE, The Hong Kong University of Science and Technology\\
Email: \{eejcshen, eejzhang, eekhaled\}@ust.hk}}
\maketitle
\begin{abstract}
Pilot contamination has been regarded as a main limiting factor of
time division duplexing (TDD) massive multiple-input\textendash{}multiple-output
(Massive MIMO) systems, as it will make the signal-to-interference-plus-noise
ratio (SINR) saturated. However, how pilot contamination will limit
the user capacity of downlink Massive MIMO, i.e., the maximum number
of admissible users, has not been addressed. This paper provides an
explicit expression of the Massive MIMO user capacity in the pilot-contaminated
regime where the number of users is larger than the pilot sequence
length. Furthermore, the scheme for achieving the user capacity, i.e.,
the uplink pilot training sequence and downlink power allocation,
has been identified. By using this capacity-achieving scheme, the
SINR requirement of each user can be satisfied and energy-efficient
transmission is feasible in the large-antenna-size (LAS) regime. Comparison
with two non-capacity-achieving schemes highlights the superiority
of our proposed scheme in terms of achieving higher user capacity.%
\footnote{This work is supported by the Hong Kong Research Grant Council under
Grant No. 610212.%
}\end{abstract}
\begin{IEEEkeywords}
Massive MIMO, user capacity, pilot contamination, pilot-aided channel
estimation, power allocation.
\end{IEEEkeywords}

\section{Introduction}

Massive MIMO is regarded as an efficient and scalable approach for
multicell multiuser MIMO implementation. By deploying base stations
(BSs) with much more antennas than active user equipments (UEs), the
asymptotic orthogonality among MIMO channels becomes valid and in
turn it makes intra- and inter-cell interference more manageable.
Hence, using simple linear precoder and detector can approach the
optimal dirty-paper coding capacity \cite{Gao11linear}. Channel side
information (CSI) at BSs plays an important role in the exploitation
of channel orthogonality. In practice, TDD operation is assumed for
CSI acquisition through uplink training. The advantage of uplink training
is that the length of pilot training sequences is proportional to
the number of active UEs rather than that of BS antennas. The training
length is fundamentally limited by the channel coherence time, which
can be short due to UEs of high mobility. It has been shown, nevertheless,
in \cite{Marzetta06} that the effect of using short training sequences
diminishes in the LAS regime. Specifically, among the poorly estimated
channels, the asymptotic orthogonality can still hold. 

However, a major problem with TDD Massive MIMO is that inevitably
the same pilot sequences will be reused in multiple cells. The channels
to UEs in different cells who share the same pilot sequence will be
collectively learned by BSs. In other words, the desired channel learned
by a BS is contaminated by undesired channels. Once this contaminated
CSI is utilized for transmitting or receiving signals, intercell interference
occurs immediately which limits the achievable SINR. This phenomenon,
known as \emph{pilot contamination}, can not be circumvented simply
by increasing the BS antenna size \cite{Marzetta10}. Several attempts
have been made to tackle this problem. In \cite{Jose11}, a sophisticated
precoding method is proposed to minimize intercell interference due
to pilot contamination. A more direct approach to pilot decontamination,
which promises to purify the polluted channel information, can be
found in \cite{Yin13}. By harnessing second-order channel statistics,
it can remove contamination from undesired channels which occupy different
angle-of-arrival intervals from the desired channel. A recent study
\cite{Mueller14} claims that pilot contamination is due to inappropriate
linear channel estimation. Hence, by using the proposed subspace-based
estimation, unpolluted CSI is obtainable. One thing to note is that
this method should be accompanied by suitable frequency reuse pattern
and power control among UEs. Also, the asymptotic effectiveness of
the last two methods in the LAS regime has been analytically presented. 

The effect of pilot contamination is usually quantified as SINR saturation
due to intercell interference. A number of studies have examined this
saturation phenomenon and the corresponding uplink or downlink throughput
\cite{Marzetta10,Ngo13EE,Hoydis13ULDL}. The former two provide analysis
in the LAS regime with a fixed number of active UEs, while in \cite{Hoydis13ULDL},
it analyzes asymptotic SINRs with a fixed ratio of the BS antenna
size to the active-UE number. All these studies lead to a similar
conclusion that the SINR will saturate with an increasing antenna
size, making system throughput limited.

So far, however, there has been little discussion about the user capacity
of TDD Massive MIMO, which is the maximum number of UEs whose SINR
requirements can be met for a given pilot sequence length. This term
``user capacity'' was originally coined for analyzing CDMA systems
\cite{Viswanath99OptSeqPow}. In this paper, we confine our discussion
to the user capacity of single-cell downlink TDD Massive MIMO, where
pilot contamination will occur once the number of UEs is greater than
the pilot sequence length. Meanwhile, we consider a more general set
of pilot sequences whose cross-correlations can range from $-1$ to
$1$. In most studies of Massive MIMO, the cross-correlations are
restricted to be $1$ or $0$. Our discussion will show that the user
capacity can be characterized by a specific region within which there
exists a capacity-achieving pilot sequence and power allocation such
that the SINR requirements are satisfied and the user data can be
energy-efficiently transmitted. This result is significant in the
sense that within the specified region we have no worries about pilot
contamination by using the proposed allocation scheme. Meanwhile,
the derivation of this result only involves simple channel estimation
and linear precoding. It means that this region can exist without
relying on advanced estimation or precoding methods. Though this region
is identified within a single cell, it sheds light on the possible
existence of the similar region in the general multicell scenario.

Notations: $\mathbb{R}$: real number, $\mathbb{Z}$: integers, $\left\Vert \cdot\right\Vert _{p}$:
$p$-norm, $\left(\cdot\right){}^{T}$: transpose, $\left(\cdot\right){}^{H}$:
Hermitian transpose, $\otimes$: Kronecker product, $\circ$: Hadamard
product, $\mathbf{I}_{N}$: $N\times N$ identity matrix, $\mathcal{CN}\left(\cdot,\cdot\right)$:
complex normal distribution, $\mathbb{E}\left[\cdot\right]$: expectation,
$\mbox{tr}\left(\cdot\right)$: trace, $\mbox{diag}\left(\cdots\right)$:
diagonal matrix, $\succ,\succeq$: vector inequalities, $\mathbf{0}$:
zero vector, $\mbox{Null}\left(\cdot\right)$: null space, $\mbox{card}\left(\cdot\right)$:
cardinality.

\section{TDD System Model}

Consider a single-cell network where a BS equipped with $M$ antennas
serves $K$ single-antenna UEs, assuming TDD operation. The BS acquires
downlink CSI through uplink pilot training. The acquired CSI will
be utilized to form linear precoding matrix for downlink spatial multiplexing.

\subsection{Uplink Training}

During the uplink training phase, each UE transmits its own pilot
sequence $\mathbf{s}_{i}\in\mathbb{R}^{\tau\times1}$, where $i$
is the UE index, $\left\Vert \mathbf{s}_{i}\right\Vert _{2}=1$, and
$\mathbf{s}_{i}^{T}\mathbf{s}_{j}=\rho_{ij}$, which is the correlation
between different training sequences. The pilot data over the block-fading
channel, synchronously received at the BS, can be expressed as

\begin{equation}
\mathbf{y}_{\left(\tau M\times1\right)}=\sum_{i=1}^{K}\mathbf{S}_{i}\mathbf{h}_{i}+\mathbf{z},
\end{equation}
where $\mathbf{S}_{i}$ is the $\tau M\times M$ matrix given by $\mathbf{S}_{i}=\mathbf{s}_{i}\otimes\mathbf{I}_{M}$,
$\mathbf{z}$ is the additive Gaussian noise distributed as $\mathcal{CN}\left(\mathbf{0},\sigma_{z}^{2}\mathbf{I}_{\tau M}\right)$,
and $\mathbf{h}_{i}$ for $1\leq i\leq K$ are identically and independently
distributed (i.i.d.) channel vectors with distribution $\mathcal{CN}\left(\mathbf{0},\mathbf{I}_{M}\right)$.
The pilot length is implicitly assumed to be less than one channel
block. This channel model is commonly assumed in Massive MIMO systems
\cite{Marzetta10,Jose11,Ngo13EE}. A simple single-user least-squares
estimate of the $i$th channel vector is given by

\begin{eqnarray}
\mathbf{\hat{h}}_{i} & = & \mathbf{P}_{\tiny{\mbox{LS}},i}\mathbf{y},\nonumber \\
 & = & \mathbf{h}_{i}+\sum_{j\neq i}^{K}\rho_{ij}\mathbf{h}_{j}+\mathbf{S}_{i}^{T}\mathbf{z},
\end{eqnarray}
where $\mathbf{P}_{\tiny{\mbox{LS}},i}=\mathbf{S}_{i}^{T}$. This
estimate indicates how the desire channel information is polluted
by undesired channels in the pilot-contaminated regime $\left(K>\tau\right)$,
where $\rho_{ij}$ may not be $0$ for some $j\neq i$.

\subsection{Downlink Transmission}

Exploiting estimated CSI at the BS, maximum ratio transmission (MRT)
precoded data are formed and simultaneously transmitted to UEs. The
received signal at the $i$th UE is

\begin{eqnarray}
\mathbf{r}_{i} & = & \mathbf{h}_{i}^{H}\left(\sum_{j=1}^{K}\mathbf{t}_{j}x_{j}\right)+w_{i},
\end{eqnarray}
where $\mathbf{t}_{i}\triangleq\nicefrac{\hat{\mathbf{h}}_{i}}{\left\Vert \hat{\mathbf{h}}_{i}\right\Vert _{2}}$
is a MRT precoding vector, $x_{i}$ denotes uncorrelated zero-mean
data with power $\mathbb{E}\left[x_{i}^{H}x_{i}\right]=P_{i}$, and
$w_{i}$ is the zero-mean noise with variance $\sigma_{w}^{2}$. In
the LAS regime ($M\gg K$), the following asymptotic results can be
applied \cite{Ngo13EE}

\[
\lim_{M\rightarrow\infty}\frac{1}{M}\mathbf{h}_{i}^{H}\mathbf{h}_{j}=\begin{cases}
0, & \mbox{if }i\neq j,\\
1, & \mbox{if }i=j.
\end{cases}
\]
Such asymptotic orthogonality has been experimentally verified in
realistic propagation environments \cite{Gao11linear}. With this
in mind, the received signal $\mathbf{r}_{i}$ can be approximated
by

\begin{equation}
\mathbf{r}_{i}\approx\sum_{j=1}^{K}\frac{M\rho_{ji}x_{j}}{\sqrt{M\left(\sum_{l=1}^{K}\rho_{il}^{2}\right)}}+w_{i},
\end{equation}
due to

\begin{eqnarray}
\mathbf{h}_{i}^{H}\mathbf{t}_{j} & = & \nicefrac{\left(\frac{\mathbf{h}_{i}^{H}\mathbf{h}_{j}}{M}+\sum_{l\neq j}^{K}\rho_{jl}\frac{\mathbf{h}_{i}^{H}\mathbf{h}_{l}}{M}+\frac{\mathbf{h}_{i}^{H}\mathbf{S}_{j}^{H}\mathbf{z}}{M}\right)}{\frac{\left\Vert \hat{\mathbf{h}}_{j}\right\Vert _{2}}{M}},\nonumber \\
 & \approx & \nicefrac{\rho_{ji}}{\frac{\left\Vert \hat{\mathbf{h}}_{j}\right\Vert _{2}}{M}},
\end{eqnarray}

\begin{eqnarray}
\frac{\left\Vert \hat{\mathbf{h}}_{i}\right\Vert _{2}}{M} & \approx & \sqrt{\frac{1}{M}\left(\sum_{j=1}^{K}\rho_{ij}^{2}\right)},
\end{eqnarray}
and some second-order results, i.e., $\lim_{M\rightarrow\infty}\nicefrac{\mathbf{h}_{j}^{H}\mathbf{S}_{i}^{T}\mathbf{z}}{M^{2}}=0$
and $\lim_{M\rightarrow\infty}\nicefrac{\mathbf{z}^{H}\mathbf{S}_{i}\mathbf{S}_{i}^{T}\mathbf{z}}{M^{2}}=0$.
The corresponding SINR is given by

\begin{eqnarray}
\mbox{SINR}_{i} & \approx & \frac{P_{i}}{\sum_{j\neq i}^{K}\rho_{ji}^{2}P_{j}},\nonumber \\
 & = & \frac{P_{i}}{\mbox{tr}\left(s_{i}^{T}\mathbf{S}\mathbf{D}\mathbf{S}^{T}s_{i}\right)-P_{i}},\label{eq: SINR approximation}
\end{eqnarray}
where $\mathbf{D}=\mbox{diag}\left(P_{1},\cdots,P_{K}\right)$, $\mathbf{S}=\left[\mathbf{s}_{1},\mathbf{s}_{2},\cdots,\mathbf{s}_{K}\right]$,
and the fact $\lim_{M\rightarrow\infty}\nicefrac{\sigma_{w}^{2}}{M}=0$
is applied. This SINR expression tells that the downlink transmission
operates in the interference-limited regime because of using a large
number of antennas. Moreover, the interference part, $\sum_{j\neq i}^{K}\rho_{ji}^{2}P_{j}$,
can not be simultaneously eliminated for every user in the pilot-contaminated
regime as non-orthogonal pilot sequences have to be used.

\section{User Capacity}

A group of UEs is said to be \emph{admissible} in the specified TDD
Massive MIMO system if there exists a feasible pilot sequence matrix
$\mathbf{S}$ and a power allocation vector $\mathbf{p}=\left[P_{1},\cdots,P_{K}\right]^{T}\succ\mathbf{0}$
such that the SINR requirements, $\mbox{SINR}_{i}\geq\gamma_{i}$
for $1\leq i\leq K$, can be jointly satisfied. A pilot sequence matrix
$\mathbf{S}$ is feasible if $\mathbf{S}\in\mathcal{S}=\left\{ \left[\mathbf{s}_{1},\mathbf{s}_{2},\cdots,\mathbf{s}_{K}\right],\mbox{ }\mathbf{s}_{i}\in\mathbb{R}^{\tau\times1}|\left\Vert \mathbf{s}_{i}\right\Vert _{2}=1\right\} $. 

In the following discussion, we will treat the approximation in (\ref{eq: SINR approximation})
as the exact SINR expression, and focus on the pilot-contaminated
regime. The proposition below gives the upper bound of the maximum
number of admissible UEs.
\begin{prop}
\label{Prop01} If $K$ UEs are admissible in the TDD Massive MIMO
system, then\label{prop: Prop01}
\end{prop}
\begin{equation}
K\leq\left[\tau\left(\sum_{i=1}^{K}1+\frac{1}{\gamma_{i}}\right)\right]^{\nicefrac{1}{2}}.\label{eq: upper bound of admissible users}
\end{equation}

\begin{proof} Making use of (\ref{eq: SINR approximation}), we have

\begin{multline}
\sum_{i=1}^{K}\frac{1+\mbox{SINR}_{i}}{\mbox{SINR}_{i}}=\sum_{i=1}^{K}\frac{1}{P_{i}}\mbox{tr}\left(s_{i}^{T}\mathbf{S}\mathbf{D}\mathbf{S}^{T}s_{i}\right),\\
=\mbox{tr}\left(\mathbf{D}^{-1}\mathbf{S}^{T}\mathbf{S}\mathbf{D}\mathbf{S}^{T}\mathbf{S}\right),\\
=\mbox{tr}\left(\mathbf{D}^{\nicefrac{-1}{2}}\mathbf{G}_{\mathbf{s}}\mathbf{D}\mathbf{G}_{\mathbf{s}}\mathbf{D}^{\nicefrac{-1}{2}}\right),\label{eq: trace inequality 01}
\end{multline}
where

\begin{eqnarray}
\mathbf{G}_{\mathbf{s}} & \triangleq & \mathbf{S}^{T}\mathbf{S},\nonumber \\
 & = & \left[\begin{array}{ccccc}
1 & \rho_{12} & \rho_{13} & \cdots & \rho_{1K}\\
\rho_{12} & 1 & \rho_{23} & \cdots & \rho_{2K}\\
\rho_{13} & \rho_{23} & 1 & \cdots & \rho_{3K}\\
\vdots & \vdots & \vdots & \ddots & \vdots\\
\rho_{1K} & \rho_{2K} & \rho_{3K} & \cdots & 1
\end{array}\right].
\end{eqnarray}
Also, we can expand the trace in (\ref{eq: trace inequality 01})
and obtain its lower bound as below

\begin{align}
 & \mbox{tr}\left(\mathbf{D}^{\nicefrac{-1}{2}}\mathbf{G}_{\mathbf{s}}\mathbf{D}\mathbf{G}_{\mathbf{s}}\mathbf{D}^{\nicefrac{-1}{2}}\right)\nonumber \\
= & K+\sum_{i=1}^{K}\sum_{j>i=1}^{K}\left(\frac{P_{i}}{P_{j}}+\frac{P_{j}}{P_{i}}\right)\rho_{ij}^{2},\nonumber \\
\geq & K+\sum_{i=1}^{K}\sum_{j>i=1}^{K}2\rho_{ij}^{2},\nonumber \\
= & \mbox{tr}\left(\mathbf{G}_{\mathbf{s}}\mathbf{G}_{\mathbf{s}}\right),
\end{align}
where the inequality is due to $\left(\nicefrac{P_{i}}{P_{j}}+\nicefrac{P_{j}}{P_{i}}\right)\geq2$.
The Gram matrix $\mathbf{G}_{\mathbf{s}}$ has an eigendecomposition
$\mathbf{U}\mathbf{D_{G}}\mathbf{U}^{T}$, where $\mathbf{U}$ is
an unitary matrix and $\mathbf{D_{G}}=\mbox{diag}\left(d_{1},\cdots,d_{K}\right)$
with $d_{1}\sim d_{\tau}>0$, $d_{\tau+1}\sim d_{K}=0$, and $\sum_{i=1}^{\tau}d_{i}=K$.
Then, we have

\begin{eqnarray}
\mbox{tr}\left(\mathbf{G}_{\mathbf{s}}\mathbf{G}_{\mathbf{s}}\right) & = & \mbox{tr}\left(\mathbf{U}\mathbf{D}_{\mathbf{G}}^{2}\mathbf{U}^{T}\right),\nonumber \\
 & = & \sum_{i=1}^{\tau}d_{i}^{2},\nonumber \\
 & \geq & \frac{1}{\tau}\left(\sum_{i=1}^{\tau}d_{i}\right)^{2}=\frac{K^{2}}{\tau}.
\end{eqnarray}
As

\begin{equation}
\sum_{i=1}^{K}\frac{1+\mbox{SINR}_{i}}{\mbox{SINR}_{i}}\leq\sum_{i=1}^{K}\frac{1+\gamma_{i}}{\gamma_{i}},\label{eq: SINR inequality}
\end{equation}
the desired inequality follows. \end{proof}

It is clear from this proposition that the number of admissible UEs
is fundamentally limited once the length of pilot sequences and the
SINR requirements are given. To offer another explanation, let's define
the normalized mean-square error seen by the $i$th UE as $\mbox{MSE}_{i}=\nicefrac{1}{\mbox{SINR}_{i}}$.
A lower bound on the sum of $\mbox{MSE}_{i}$ is given by

\begin{eqnarray}
\sum_{i=1}^{K}\mbox{MSE}_{i} & = & \mbox{tr}\left(\mathbf{D}^{-1}\mathbf{S}^{T}\mathbf{S}\mathbf{D}\mathbf{S}^{T}\mathbf{S}\right)-K,\nonumber \\
 & \geq & \frac{K^{2}}{\tau}-K.\label{eq: MSE inequality}
\end{eqnarray}
Appealing to (\ref{eq: SINR inequality}) and (\ref{eq: MSE inequality})
leads to the same result as Proposition \ref{Prop01}, which links
the bound on the admissible UEs to the bound on the sum of mean-square
errors.

Proposition \ref{prop: Prop01} provides an upper bound for the user
capacity. The next question to ask is whether $K$ UEs are admissible
if the inequality (\ref{eq: upper bound of admissible users}) is
satisfied, i.e., the achievability issue. In other words, once the
UE number is less than or equal to the upper bound, we wonder if there
exists a set of $\mathbf{S}\in\mathcal{S}$ and $\mathbf{p}\succ\mathbf{0}$,
fulfilling the SINR requirements. The following section will show
that the answer to this question is positive.

\section{Capacity-Achieving Pilot Sequence and Power Allocation}

Validating the converse of Proposition \ref{prop: Prop01} requires
to identify $\mathbf{S}$ and $\mathbf{p}$ with which the given SINR
requirements are met. The constraints, $\mbox{SINR}_{i}\geq\gamma_{i}$
for $1\leq i\leq K$, can be recast as $\mathbf{A}\mathbf{p}\succeq\mathbf{0}$,
where

\begin{equation}
\mathbf{A}=\left[\begin{array}{cccc}
\frac{1}{\gamma_{1}} & -\rho_{12}^{2} & \cdots & -\rho_{1K}^{2}\\
-\rho_{21}^{2} & \frac{1}{\gamma_{2}} & \cdots & -\rho_{2K}^{2}\\
\vdots & \vdots & \ddots & \vdots\\
-\rho_{K1}^{2} & -\rho_{K2}^{2} & \cdots & \frac{1}{\gamma_{K}}
\end{array}\right].
\end{equation}
If there exists $\mathbf{p}\succ\mathbf{0}$ in the null space of
$\mathbf{A}$, any linearly scaled $\alpha\mathbf{p}$ for $\alpha>0$
is still a valid solution of the problem. In this case, the total
transmission power can be made fairly small. The result is due to
downlink Massive MIMO being in the interference-limited regime. The
specific definition of a valid set of $\mathbf{S}$ and $\mathbf{p}$
is given below.
\begin{defn}
The set of a pilot sequence matrix $\mathbf{S}\in\mathcal{S}$ and
a power allocation vector $\mathbf{p}\succ\mathbf{0}$ is said to
be valid if
\[
\mathbf{p}\in\mbox{Null}\left(\mathbf{T}-\mathbf{G}_{\mathbf{s}}^{T}\circ\mathbf{G}_{\mathbf{s}}\right),
\]
where $\mathbf{T}=\mbox{diag}\left(1+\frac{1}{\hat{\gamma}_{1}},\cdots,1+\frac{1}{\hat{\gamma}_{K}}\right)$,
$\hat{\gamma}_{i}\geq\gamma_{i}$ for $1\leq i\leq K$, and $\mbox{rank}\left(\mbox{Null}\left(\mathbf{T}-\mathbf{G}_{\mathbf{s}}^{T}\circ\mathbf{G}_{\mathbf{s}}\right)\right)>1$.
\end{defn}
In this definition, $\mathbf{A}=\mathbf{T}-\mathbf{G}_{\mathbf{s}}^{T}\circ\mathbf{G}_{\mathbf{s}}$
when $\hat{\gamma}_{i}=\gamma_{i}$. Also, if the achieved $\mbox{SINR}_{i}=\hat{\gamma}_{i}$
is higher than the required $\gamma_{i}$, we regard the corresponding
$\mathbf{S}$ and $\mathbf{p}$ as valid. The following proposition
will specify a region where a valid pilot sequence and power allocation
can exist.
\begin{prop}
If 

\begin{equation}
\sum_{i=1}^{K}\left(\frac{\gamma_{i}}{1+\gamma_{i}}\right)\leq\tau,\label{eq: UC for General SINRs}
\end{equation}
 then $K\leq\left[\tau\left(\sum_{i=1}^{K}1+\frac{1}{\gamma_{i}}\right)\right]^{\nicefrac{1}{2}}$
and there exists a valid pilot sequence and power allocation. \label{prop:Prop02}
\end{prop}
\begin{proof} The Cauchy\textendash{}Schwarz inequality gives

\begin{eqnarray*}
\sum_{i=1}^{K}\left(1+\frac{1}{\gamma_{i}}\right) & \geq & \frac{K^{2}}{\sum_{i=1}^{K}\left(\frac{\gamma_{i}}{1+\gamma_{i}}\right)},\\
 & \geq & \frac{K^{2}}{\tau},
\end{eqnarray*}
which is equivalent to (\ref{eq: upper bound of admissible users})
and proves the first part of the statement. Before going on to the
second part, a definition and a lemma to be utilized later are provided.
\begin{defn}
Given $\mathbf{x},\mbox{ }\mathbf{y}\in\mathbb{R}^{N}$, $\mathbf{x}$
majorizes $\mathbf{y}$ if

\[
\sum_{k=1}^{n}x_{\left[k\right]}\geq\sum_{k=1}^{n}y_{\left[k\right]},\mbox{ for }n=1,\cdots,N,
\]
where $x_{\left[k\right]}$ and $y_{\left[k\right]}$ are respectively
the elements of $\mathbf{x}$ and $\mathbf{y}$ in decreasing order.\end{defn}
\begin{lem}
\cite[Theorem 9.B.2]{Marshall10} Given $\mathbf{x},\mbox{ }\mathbf{y}\in\mathbb{R}^{N}$,
if $\mathbf{x}$ majorizes $\mathbf{y}$ and $\sum_{k=1}^{N}x_{\left[k\right]}=\sum_{k=1}^{N}y_{\left[k\right]}$,
then there exists a real symmetric matrix $H$ with diagonal elements
$y_{\left[k\right]}$ and eigenvalues $x_{\left[k\right]}$.\label{lem: Majorization}
\end{lem}
First consider the case of $\sum_{i=1}^{K}\left(\frac{\gamma_{i}}{1+\gamma_{i}}\right)=\tau$. 

Given that the $1\times K$ vector of eigenvalues $\mathbf{e}=\left[\lambda_{1},\cdots,\lambda_{\tau},0\cdots,0\right]^{T}$
majorizes $\mathbf{p}$ and $\sum_{i=1}^{\tau}\lambda_{i}=\sum_{i=1}^{K}P_{i}$,
because of Lemma \ref{lem: Majorization}, there exists a real symmetric
matrix $\mathbf{H}=\mathbf{Q}\Lambda\mathbf{Q}^{T}$, where the vector
of diagonal entries of $\mathbf{H}$ is equal to $\mathbf{p}$, $\Lambda=\mbox{diag}\left(\lambda_{1},\cdots,\lambda_{\tau},0\cdots,0\right)$,
and the orthogonal matrix $\mathbf{Q}$ can be presented as

\[
\left[\begin{array}{cc}
\mathbf{V}_{K\times\tau} & \tilde{\mathbf{V}}_{K\times\left(K-\tau\right)}\end{array}\right].
\]
The approach to constructing $\mathbf{Q}$ as well as $\mathbf{H}$
is provided in \cite[Sec. IV-A]{Viswanath99}. Define that

\begin{equation}
\mathbf{S}\triangleq\Sigma^{\nicefrac{1}{2}}\mathbf{V}^{T}\mathbf{D}^{\nicefrac{-1}{2}},
\end{equation}
where $\Sigma=\mbox{diag}\left(\lambda_{1},\cdots,\lambda_{\tau}\right)$.
Then, $\mathbf{S}\in\mathcal{S}$ is true as the diagonal entries
of $\mathbf{S}^{T}\mathbf{S}=\mathbf{D}^{\nicefrac{-1}{2}}\mathbf{H}\mathbf{D}^{\nicefrac{-1}{2}}$
are equal to $1$. Moreover, we have

\begin{equation}
\mathbf{S}\mathbf{D}\mathbf{S}^{T}=\Sigma.\label{eq: SDS^T}
\end{equation}
Let's specify that

\begin{equation}
\lambda_{1}=\cdots=\lambda_{\tau}=\frac{\sum_{i=1}^{K}P_{i}}{\tau},
\end{equation}
and

\begin{equation}
P_{i}=c\frac{\gamma_{i}}{1+\gamma_{i}},\mbox{ for some }c>0.
\end{equation}
It can be verified that $\mathbf{e}$ majorizes $\mathbf{p}$ since
for $1\leq i\leq\tau$,

\begin{eqnarray*}
\lambda_{i} & = & \frac{c\sum_{i=1}^{K}\frac{\gamma_{i}}{1+\gamma_{i}}}{\tau},\\
 & = & c,\\
 & > & \max\left\{ P_{k},\mbox{ for }1\leq k\leq K\right\} ,
\end{eqnarray*}
where the second equality is due to the case under consideration.

Next we will check if the SINR requirements are satisfied by using
such pilot sequence and power allocation. Making use of (\ref{eq: SDS^T}),
we have

\begin{eqnarray*}
\mbox{SINR}_{i} & = & \frac{P_{i}}{\mbox{tr}\left(s_{i}^{T}\Sigma s_{i}\right)-P_{i}},\\
 & = & \frac{c\frac{\gamma_{i}}{1+\gamma_{i}}}{c-c\frac{\gamma_{i}}{1+\gamma_{i}}},\\
 & = & \gamma_{i},\mbox{ }\forall i=1,\cdots,K.
\end{eqnarray*}
Based on this result, it can be easily shown that $\mathbf{p}\in\mbox{Null}\left(\mathbf{T}-\mathbf{G}_{\mathbf{s}}^{T}\circ\mathbf{G}_{\mathbf{s}}\right)$
where the diagonal matrix $\mathbf{T}=\mbox{diag}\left(1+\nicefrac{1}{\gamma_{1}},\cdots,1+\nicefrac{1}{\gamma_{K}}\right)$.

Now we turn to the case of $\sum_{i=1}^{K}\left(\nicefrac{\gamma_{i}}{1+\gamma_{i}}\right)<\tau$.
As $f\left(x\right)=\nicefrac{x}{1+x}$ is monotonically increasing
for $x>0$, there exists a set $\{\hat{\gamma}_{i}\geq\gamma_{i}$
for $1\leq i\leq K\}$ such that $\sum_{i=1}^{K}\left(\nicefrac{\hat{\gamma}_{i}}{1+\hat{\gamma}_{i}}\right)=\tau$.
At the same time, $K\leq\left[\tau\left(\sum_{i=1}^{K}1+\nicefrac{1}{\hat{\gamma}_{i}}\right)\right]^{\nicefrac{1}{2}}$
holds. By exploiting the previous result, we can find a valid set
of $\mathbf{S}\in\mathcal{S}$ and $\mathbf{p}\succ\mathbf{0}$ for
which $\mathbf{T}=\mbox{diag}\left(1+\nicefrac{1}{\hat{\gamma}_{1}},\cdots,1+\nicefrac{1}{\hat{\gamma}_{K}}\right)$.
\end{proof}

An explanation of the constraint, $\sum_{i=1}^{K}\left(\nicefrac{\gamma_{i}}{1+\gamma_{i}}\right)\leq\tau$,
is as follows. The UE with a high SINR requirement should be allocated
with a pilot sequence which is orthogonal to others. Overall, only
$\tau$ such assignments are allowed in the system. Another thing
to note is that the given proof is constructive, within which the
method of obtaining pilot sequences and allocating sequences and powers
to UEs can be found.

A corollary which follows from Propositions \ref{prop: Prop01} and
\ref{prop:Prop02} is provided below.
\begin{cor}
\label{Corollary01}Given the identical SINR requirement $\gamma$,
$K$ UEs are admissible in the TDD Massive MIMO system if and only
if 
\end{cor}
\begin{equation}
K\leq\left(1+\frac{1}{\gamma}\right)\tau.\label{eq: User Capacity for Identical SINRs}
\end{equation}

Unlike (\ref{eq: upper bound of admissible users}), the right-hand
side of (\ref{eq: User Capacity for Identical SINRs}) does not depend
on $K$, providing an explicit upper bound of admissible UEs. However,
this is only for identical SINR requirements. In the general case,
(\ref{eq: upper bound of admissible users}) and (\ref{eq: UC for General SINRs})
do not provide upper bounds of this kind. In order to have a consistent
interpretation of the user capacity in the general case, we intend
to characterize the user capacity as the \emph{admissible region}
$R_{\tiny{\mbox{UC}}}=\{\gamma_{1\sim K}\in\mathbb{R}^{+}|\mbox{ }\sum_{i=1}^{K}\nicefrac{\gamma_{i}}{1+\gamma_{i}}\leq\tau\}$.
It means that once the SINR requirements are located within $R_{\tiny{\mbox{UC}}}$,
the corresponding $K$ UEs are admissible. When it comes to identical
SINR requirements, this region maintains the same structure, that
is $R_{\tiny{\mbox{UC}}}=\left\{ \gamma_{1\sim K}=\gamma\in\mathbb{R}^{+}|\nicefrac{K\gamma}{1+\gamma}\leq\tau\right\} $.
Later on, this characterization will be utilized to evaluate different
joint pilot sequence design and power allocation schemes, i.e., different
combinations of $\mathbf{S}$ and $\mathbf{p}$. 

According to the present analytical results, the following remarks
can be made.
\begin{enumerate}
\item The valid pilot sequence and power allocation used in Proposition
\ref{prop:Prop02} is also referred to as the capacity-achieving pilot
sequence and power allocation. It means that any $K$ UEs having the
SINR requirements within $\mathbf{R}_{\tiny{\mbox{UC}}}$ can be admitted
by using this allocation. In the next section, it will be shown that
other non-capacity-achieving schemes can not guarantee this.
\item When using the capacity-achieving pilot sequence and power allocation,
the converse of Proposition \ref{prop:Prop02} can be shown to be
true.
\item In the case of identical SINR requirements, the pilot sequences in
use are called Welch bound equality (WBE) sequences (with properties:
$\mathbf{S}\in\mathcal{S}$, $\mathbf{S}\mathbf{S}^{T}=\frac{K}{\tau}\mathbf{I}_{\tau}$,
and $\rho_{ij}^{2}=\nicefrac{\left(K-\tau\right)}{\left(K-1\right)\tau}$
for $i\neq j$) \cite{Ulukus01,Waldron03}.
\item Generally, the regime $K\leq\left[\tau\left(\sum_{i=1}^{K}1+\nicefrac{1}{\gamma_{i}}\right)\right]^{\nicefrac{1}{2}}$
while $\sum_{i=1}^{K}\left(\nicefrac{\gamma_{i}}{1+\gamma_{i}}\right)>\tau$,
has not been characterized, in which the existence of a valid pilot
sequence and power allocation is unknown. However, when all the SINR
requirements are the same, $K\leq\left[\tau\left(\sum_{i=1}^{K}1+\nicefrac{1}{\gamma}\right)\right]^{\nicefrac{1}{2}}$
implies $\sum_{i=1}^{K}\left(\nicefrac{\gamma}{1+\gamma}\right)\leq\tau$.
In this special case, the setting, $\sum_{i=1}^{K}\left(\nicefrac{\gamma}{1+\gamma}\right)>\tau$,
is not of interest.
\end{enumerate}

\section{Comparison with Non-Capacity-Achieving Schemes}

The superiority of the proposed capacity-achieving allocation over
other existing schemes will be presented in this section. Let's first
define two pilot sequence allocation schemes which are independent
of the SINR requirements. Meanwhile, in both schemes, the transmit
power $P_{i}$ allocated to the $i$th UE is $\nicefrac{c\gamma_{i}}{1+\gamma_{i}}$
for some $c>0$.
\begin{enumerate}
\item WBE Scheme: Pilot sequences in use are the WBE sequences.
\item Finite Orthogonal Sequence (FOS) Scheme: Given a pilot sequence length
$\tau$, only $\tau$ orthogonal pilot sequences will be repeatedly
used in the pilot-contaminated regime. Assume that $K=q\tau+r$ where
$q,r\in\mathbb{Z}$ and $0\leq r<\tau$. Each pilot sequence $s_{i}$
is used by a collection $E_{i}$ of UEs. Let $\mbox{card}\left(E_{i}\right)=q+1$
for $1\leq i\leq r$, $\mbox{card}\left(E_{i}\right)=q$ for $r+1\leq i\leq\tau$,
and $E_{i}\cap E_{j}=\emptyset$ for $i\neq j$. 
\end{enumerate}
The following lemmas will show the potential reduction of the user
capacity when the WBE and FOS scheme are applied to the case of general
SINR constraints.
\begin{lem}
The general SINR requirements $\gamma_{i}$ are satisfied by using
the WBE scheme if and only if

\begin{equation}
\sum_{i=1}^{K}\left(\frac{\gamma_{i}}{1+\gamma_{i}}\right)\leq\min\left\{ \tau,\mbox{ }\kappa-\left(\kappa-1\right)\left(\frac{\gamma_{\tiny{\mbox{max}}}}{1+\gamma_{\tiny{\mbox{max}}}}\right)\right\} ,\label{eq: WBE user bound}
\end{equation}
where $\kappa=\frac{\left(K-1\right)\tau}{\left(K-\tau\right)}$ and
$\gamma_{\tiny{\mbox{max}}}=\max\left\{ \gamma_{i},\mbox{ }1\leq i\leq K\right\} $.\label{lem:Lemma02}
\end{lem}
\begin{proof} Please refer to \cite{Shen14_MassiveMIMO_TWC} due
to space limitations. \end{proof}
\begin{lem}
The general SINR requirements $\gamma_{i}$ are satisfied by using
the FOS scheme if and only if\label{lem:Lemma03}

\begin{equation}
\sum_{k\in E_{i}}\left(\frac{\gamma_{k}}{1+\gamma_{k}}\right)\leq1,\mbox{ for }1\leq i\leq\tau,\label{eq: FOS constraints}
\end{equation}
and

\begin{equation}
\sum_{i=1}^{K}\left(\frac{\gamma_{i}}{1+\gamma_{i}}\right)\leq\tau.\label{eq: FOS constraints 02}
\end{equation}

\end{lem}
\begin{proof} Similar to the proof of Lemma \ref{lem:Lemma02}. \end{proof}

\begin{figure}[t]
\includegraphics[width=9cm,height=6cm]{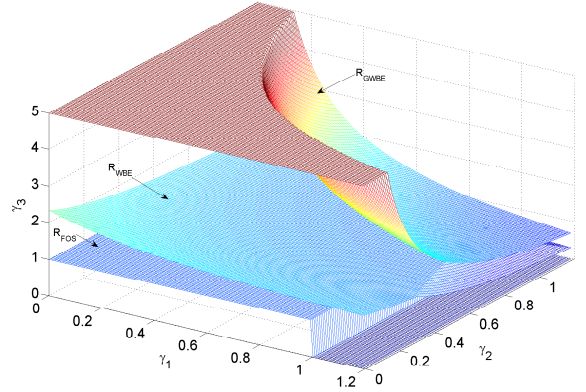}

\caption{Upper boundaries of admissible regions for the GWBE, WBE, and FOS
schemes.\label{fig: GWBE & FOS Feasible Regions}}
\end{figure}

To verify the results in Proposition \ref{prop:Prop02} and in Lemmas
\ref{lem:Lemma02} and \ref{lem:Lemma03}, we consider a pilot-contaminated
Massive MIMO system with $K=6$ and $\tau=3$. By fixing certain SINR
requirements $\left\{ \gamma_{4}=\gamma_{5}=\gamma_{6}=1\right\} $,
we look into admissible regions of the remaining SINR requirements
given by

\begin{equation}
R_{\mbox{\tiny{GWBE}}}=\left\{ \gamma_{1\sim3}\in\mathbb{R}^{+}|\mbox{ }\sum_{i=1}^{3}\nicefrac{\gamma_{i}}{1+\gamma_{i}}\leq\nicefrac{3}{2}\right\} ,
\end{equation}
 
\begin{multline}
R_{\mbox{\tiny{WBE}}}=R_{\mbox{\tiny{GWBE}}}\cap\left\{ \gamma_{1\sim3}\in\mathbb{R}^{+}|\mbox{ }\right.\\
\left.\sum_{j=1}^{3}\nicefrac{\gamma_{j}}{1+\gamma_{j}}\leq\left(\nicefrac{7}{2}-\nicefrac{4\gamma_{i}}{1+\gamma_{i}}\right),\mbox{ for }1\leq i\leq3\right\} ,
\end{multline}
and

\begin{equation}
R_{\mbox{\tiny{FOS}}}=R_{\mbox{\tiny{GWBE}}}\cap\left\{ \gamma_{1\sim3}\in\mathbb{R}^{+}|\mbox{ }\gamma_{i}\leq1,\mbox{ for }1\leq i\leq3\right\} ,
\end{equation}
for the GWBE, WBE, and FOS schemes. Note that it is implicitly assumed
that $E_{1}=\left\{ \mbox{UE}_{1},\mbox{UE}_{4}\right\} $, $E_{2}=\left\{ \mbox{UE}_{2},\mbox{UE}_{5}\right\} $,
and $E_{3}=\left\{ \mbox{UE}_{3},\mbox{UE}_{6}\right\} $ for the
FOS scheme. The upper boundaries of these regions in the positive
orthant are plotted in Fig. \ref{fig: GWBE & FOS Feasible Regions}.
For the GWBE scheme, an extra restriction $\gamma_{3}=\mbox{min}\left\{ \gamma_{3},5\right\} $
is placed as the admissible $\gamma_{3}$ can go to infinity. It can
be observed that the boundary surface of $R_{\mbox{\tiny{GWBE}}}$
lies well above those of $R_{\mbox{\tiny{WBE}}}$ and $R_{\mbox{\tiny{FOS}}}$.
This implies that $R_{\mbox{\tiny{GWBE}}}$ contains more admissible
points than $R_{\mbox{\tiny{WBE}}}$ and $R_{\mbox{\tiny{FOS}}}$,
so more general SINR constraints $\gamma_{1\sim3}$ can be met in
the GWBE scheme.

To explore the effects of having different numbers of UEs, Fig. \ref{fig: SINR vs UEs}
plots the achievable SINR versus the number of UEs given the fixed
$\tau=3$. For the FOS scheme, the grouping among UEs for any given
$K$ is assumed to be optimal in the sense of maximizing the achievable
SINR. It can be observed that increasing $K$, making pilot contamination
more serious, leads to decreasing achievable SINRs for all three schemes.
Our proposed GWBE scheme, however, attains relatively higher SINRs
compared with the WBE and FOS schemes. Interestingly, the WBE scheme
does not always outperform the FOS scheme for $K<7$, but does so
for $K\geq7$. This highlights that the GWBE scheme exhibits a consistent
superiority over the FOS scheme compared with the WBE scheme.

By specifying the SINR-requirement pattern of $K=3l$ UEs, how many
UEs are admissible for a given pilot length is depicted in Fig. \ref{fig: UEs versus Tau}.
It can be observed that the number of admissible UEs scales almost
linearly with the pilot length whatever scheme is adopted. This linear
relationship directly demonstrates how the user capacity is limited
by the pilot length. Also shown in the same figure, the GWBE scheme,
without doubt, substantially outperforms the other two schemes in
terms of admitting more UEs. In addition, two non-capacity-achieving
schemes exhibit comparable user capacities especially at short pilot
lengths.

\begin{figure}[t]
\includegraphics[width=9cm,height=6cm]{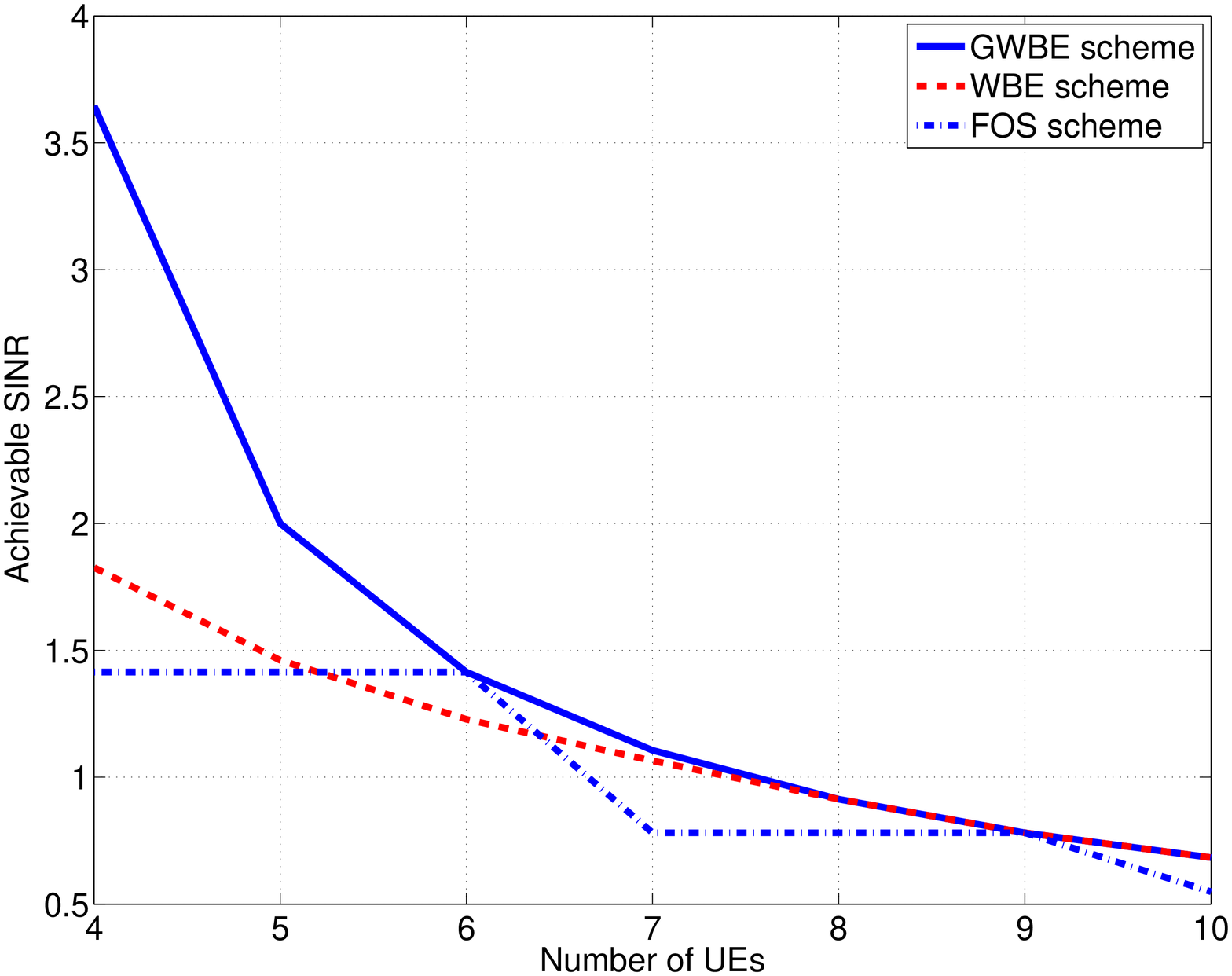}

\caption{Achievable SINR versus the number $K$ of UEs for the GWBE, WBE, and
FOS schemes given a fixed SINR-requirement pattern, that is $\left\{ \gamma_{1}=\gamma_{2}=\gamma_{3}=\gamma,\mbox{ }\gamma_{4}=\cdots=\gamma_{K}=\nicefrac{\gamma}{2}\right\} $.
\label{fig: SINR vs UEs}}
\end{figure}

\begin{figure}[t]
\noindent \begin{centering}
\includegraphics[width=9cm,height=6cm]{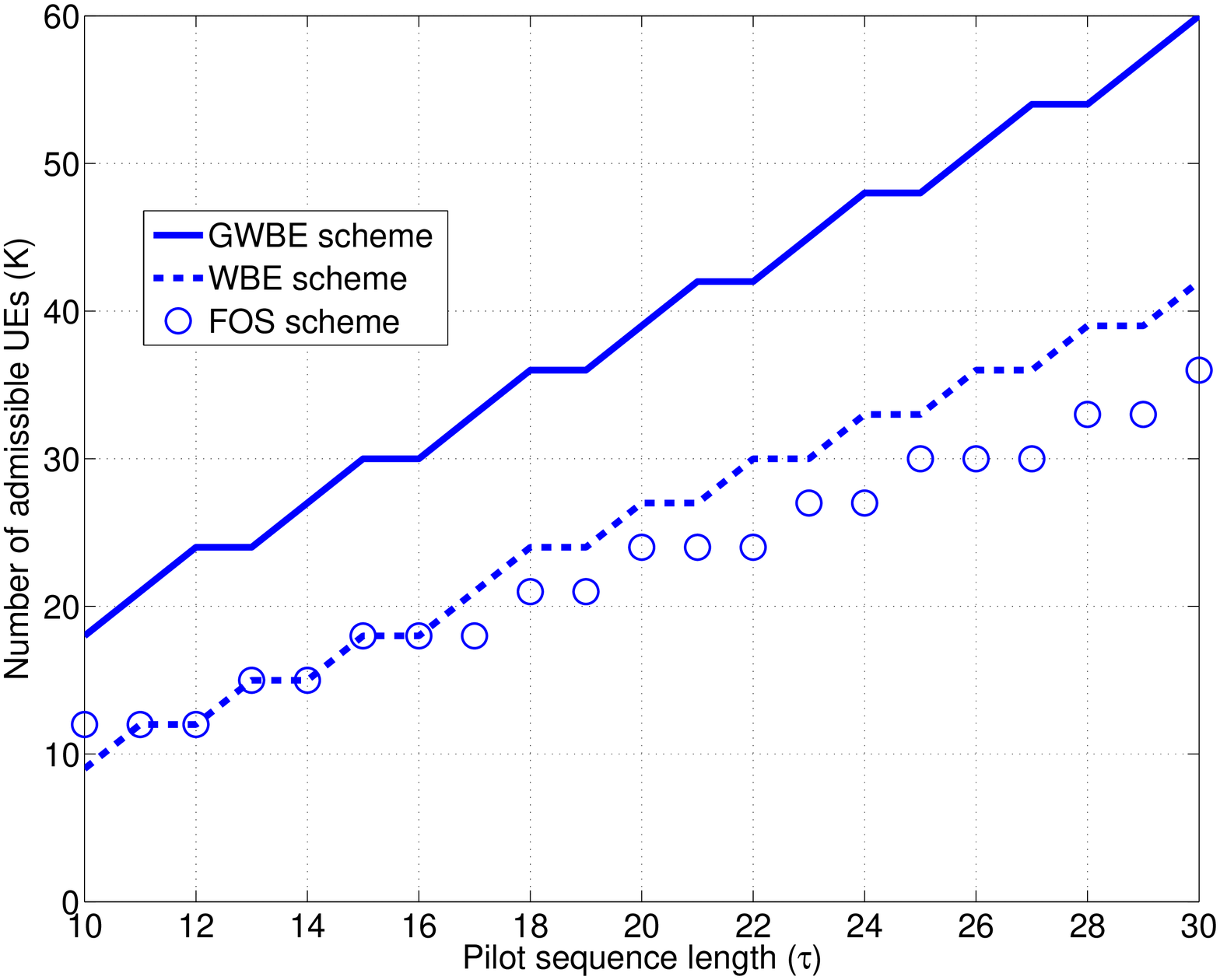}
\par\end{centering}

\caption{Number of admissible UEs versus pilot sequence length for the GWBE,
WBE, and FOS schemes, given a fixed SINR-requirement pattern, that
is $\left\{ \gamma_{1\sim l}=\nicefrac{1}{3},\mbox{ }\gamma_{\left(l+1\right)\sim2l}=1,\mbox{ }\gamma_{\left(2l+1\right)\sim3l}=3\right\} $.\label{fig: UEs versus Tau}}
\end{figure}

\section{Conclusions}

This paper has investigated the user capacity of downlink TDD Massive
MIMO systems in the pilot-contaminated regime. The necessary condition
for admitting a group of UEs with general SINR requirements has been
provided. It shows an intrinsic capacity upper bound due to the limited
length of pilot sequences. Meanwhile, the capacity-achieving pilot
sequence and power allocation, which can achieve the identified user
capacity and satisfy the SINR requirements, has been proposed and
compared with the non-capacity-achieving WBE and FOS schemes. The
results of this study indicate that the capacity-achieving allocation
is necessary for the purpose of enhancing the user capacity.

\bibliographystyle{IEEEtran}
\bibliography{MassiveMIMO}

\end{document}